\documentclass{amsart}

\usepackage{amsfonts,amssymb,amscd,amsmath,enumerate,verbatim,newlfont,calc}
\usepackage{graphicx}
\usepackage{color}
\usepackage{amsfonts}

\newtheorem{theorem}{Theorem}[section]

\newtheorem{definition}[theorem]{Definition}
\newtheorem{example}[theorem]{Example}

\begin{document}
\title{\textbf{CONNECTION STABILITY OF VECTOR FIELDS WITH ASTROPHYSICAL APPLICATION TO THE EINSTEIN-VLASOV SYSTEM}}
\author[Mohammadreza Molaei and Christian Corda]{\textbf{$^{1}$Mohammadreza Molaei and $^{2}$Christian Corda}\\
$^{1}$Department of Pure Mathematics, Shahid Bahonar University of
Kerman, Kerman,, Iran. E-mail: mrmolaei@uk.ac.ir\\ $^{2}$SUNY
Polytechnic Institute, 13502 Utica, New York, USA, and International Institute for
Applicable Mathematics and Information Sciences, B. M. Birla
Science Centre, Adarshnagar, Hyderabad 500063, India.\\ E-mail:
cordac.galilei@gmail.com} \maketitle

\begin{abstract}
In this paper we present a method for considering the stability of
smooth vector fields on a smooth manifold which may not be
compact. We show that these kind of stability which is called
"connection stability" is equivalent to the structural stability
in the case of compact manifolds. We prove if $X$ is a connection
stable vector field, then any multiplication of it by a nonzero
scalar is also a connection stable vector field. We present  an
example of a connection stable vector field on a noncompact
manifold, and we also show that harmonic oscillator is not a
connection stable vector field. We present a technique to prove a
class of vector fields are not connection stable. As a concrete
physical example, we will apply the analysis to the
Einstein-Vlasov system in an astrophysical context in which we 
propose an approach that could, in principle and partially, help 
to understand the problem of galaxy rotation curves.
\end{abstract}
\maketitle \noindent {\bf Keywords:} { Connection stability;
Structural stability; Einstein-Vlasov system; Vector field, Noncompact manifold }\\
\noindent \textbf{Statements and Declarations}:  The authors
declare no competing interests.
\section{Introduction}
Stability of vector fields on compact manifolds has been
considered by many researchers from different viewpoints, and many
papers have been published on this topic \cite{A1, B, B1, C, S}.
In fact the compact condition of the ambient manifold let us to
define a metric on its vector fields. More precisely when $M$ is a
$C^{r+1}$ ($r\geq 1$) compact manifold then we can find an open
cover $\{W_{1},\cdots, W_{s}\}$ for it so that each $W_{i}$ is a
subset of $V_{i}$ such that $(V_{i}, \varphi _{i})$ is a local
chart which $\varphi _{i}(V_{i})$ is a ball of radius $2$ and
$\varphi _{i}(W_{i})$ is a ball of radius $1$ ($B(1)$) around
origin.  A $C^{r}$ vector field $X$ in the tangent
bundle of $M$ can be denoted by $X=X^{\gamma}\partial_{\gamma}$
($\partial_{\gamma}$ also denoted by $\frac{\partial}{\partial
x_{\gamma}}$), where each $X^{\gamma} $ as a component of $X$ is a
function from $M$ to $R$. For each   $X^{\gamma} $ a $C^{r}$ norm
is defined by
$$||X^{\gamma}||_{r}=$$ $$\displaystyle \max_{i} \sup \{ ||X^{\gamma i}(u)||,
||dX^{\gamma i}(u)||,\cdots, ||d^{r}X^{\gamma i}(u)|| ~:~
X^{\gamma i}=X^{\gamma}o\varphi_{i} ^{-1}, ~u\in B(1)\}.$$ Two
vector fields on $M$ are called $C^{r}$ close if their components
are $C^{r}$ close with the former norm \cite{PD}. By this norm we
say that a vector field $X$ on $M$ is structurally stable if there
is a neighborhood of it such that each $Y$ belonging to this
neighborhood is topologically conjugate to it. This definition has
no meaning when $M$ is not compact.
\\In this work we
present a method to consider the stability of vector fields on
noncompact manifolds.
Of course the method is also applicable for compact manifolds.\\
In
2020 we present a method for answering to the following question \cite{MG}:\\
 Suppose we have
 a system of ordinary differential equations on a smooth manifold
$M$, can we find a linear connection $\nabla $ on $M$ such that
each orbit of this system be a geodesic in $(M, \nabla)$?\\ For
example if we have a flow in the nature  can we find a connection
to create such flow as its geodesics (see Figure \ref{fig1})?
Here flow means: the evolution operators created by
the integral curves of a vector field. The positive answer to
this
question have been done by the following method \cite{MG}:\\
If $\alpha $ is an integral curve of a vector field  $X$ passing
through a point $p\in M$, then we have
 \begin{equation}\label{eq3}
          \left\{
         \begin{array}{rl}
         \dot{\alpha}(t)=X(\alpha(t)) \\
         \alpha(0)=p \hspace{1cm}
         \end{array} \right..
         \end{equation}
 We find the connection components $\Gamma^{j}_{ik}$  by solving the
 system
\begin{equation}\label{eq4}\frac{d X^{j}}{dt}(t)+{X}^{k}(t){X}^{i}(t)\Gamma^{j}_{ik}(\alpha
(t))=0,
\end{equation}
by removing the subset of the phase space which a component of the
vector field $X$ is zero on it. For the proof we see
that if $ \nabla$ is a connection with the  Christoffel symbols
$\Gamma^{j}_{ik}$ then $\beta $ is a geodesic of it if $\beta $ is
a solution of the geodesic equation \cite{MG}
$\frac{d^{2}}{dt^{2}}\beta^{j}(t)+\dot{\beta}^{k}(t)\dot{\beta}^{i}(t)\Gamma^{j}_{ik}(t)=0$.
For finding connection,  in the former equation we replace $\beta$
with the integral curve $\alpha$ of a given vector field $X$, and
we have
\begin{equation}\label{eq10}\frac{d\dot{\alpha}^{j}}{dt}(t)+\dot{\alpha}^{k}(t)\dot{\alpha}^{i}(t)\Gamma^{j}_{ik}(\alpha
(t))=0. \end{equation} By equation \ref{eq3}   we  replace
$\dot{\alpha}^{k}(t)$ with ${X}^{k}(t)$, and we deduce the
following equation:
$$\frac{d X^{j}}{dt}(t)+{X}^{k}(t){X}^{i}(t)\Gamma^{j}_{ik}(\alpha
(t))=0.$$ In these equations the unknown variables are
$\Gamma^{j}_{ik}$. In fact if $m$ is the dimension of $M$ then we
have $m$ algebraic equation with $m^{3}$ unknown variables. So we
can find many different connections. For this reason for a given
$(i,k)$ we take $\Gamma^{j}_{ik}=-\frac{\frac{d
X^{j}}{dt}(t)}{{X}^{k}(t){X}^{i}(t)},$ for
$j=1,2,\cdots , m,$ and we take $\Gamma^{j}_{rs}=0 $ for
$(r,s)\neq (i,k)$. This means that for a given $(i,k)$ we deduce a
connection. So  $m^{2}$ connections can be calculated by changing
$(i,k)$. We denote the Christoffel symbols of
these $m^{2}$ resulted connections from the vector field $X$ by
$\Gamma^{jX}_{ik}$. One must pay attention to this
point that $\Gamma^{jX}_{ik}$ are not  Christoffel symbols of one
connection. In fact they are
 Christoffel symbols of $m^{2}$ connections which may not be zero.\\
 In the next section we make use of these linear connections to define
 {\it connection stability}. We show that on compact manifolds
 connection stability is equivalent to structural stability. We
 present an example of a connection stable vector field on a noncompact manifold, and we prove that if $X$ is a connection stable
 vector field, then $cX$ is a connection stable vector field for
 each $c\in \mathbb{R}\setminus \{0\}$. In the third section
 the connection stability of  harmonic oscillator is
 considered. We see that it is not a connection stable vector
 field. By putting a condition on the Christoffel symbols of a
 vector fields we deduce a class of vector fields on $R^{3}$ which are not
 connection stable (see Theorem \ref{Th2}).
 \begin{figure}\label{fig1}
\includegraphics[scale=.65]{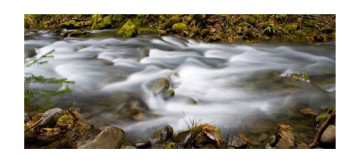}
\caption{A nature flow which has been taken from
$www.conservationgateway.org.$} \label{fig1}
\end{figure}
In section $4$ we consider a vector field deduced from
Einstein-Vlasov system and we show that it is not a connection
stable vector field.
\section{NEW CONCEPT OF STABILITY}
We assume $M$ is a smooth manifold and $V(M)$ is the set of smooth
vector fields on $M$ with finite singularities. If $X\in V(M)$ we
denote the finite set of it's singularities by $S(X),$ and we
denote it's cardinality by  $|S(X)|.$
\begin{definition}\label{def1} We say that a vector field $X\in V(M)$ is a connection stable
vector field  in the set of vector fields on $M$, if there is
$\epsilon
>0$  such that each $Y\in V(M)$ with the following two conditions is topologically conjugate to $X$. The conditions of $Y$ are $|S(Y)|=|S(X)|$, and $|\Gamma^{jY}_{ik}(r)-\Gamma^{jX}_{ik}(r)|<\epsilon $ for all
$r\in M\setminus (S(Y)\cup S(X)).$
\end{definition}
One must pay attention to this point that the former
definition is independent of the selected local chart of  the
atlas of $M$, because for a given $(i,k)$ the resulted connection
$\nabla^{(i,k)}$ with the  Christoffel symbols $$
         \Gamma^{j}_{rs}= \left\{
         \begin{array}{rl}
         \Gamma^{jX}_{ik}~if~(r,s)=(i,k) \\
        0~ ~ ~ ~ ~if~(r,s)\neq (i,k)
         \end{array} \right..
         $$ is independent
of the selected local chart, and it only depends on the atlas of
$M$. More precisely it determines by a partition of unity subordinate to the atlas of $M$.  $\Gamma^{jX}_{ik}$ is calculated via the components of the vector field $X$ which are  not depend on the selected coordinate. \\
We recall that two smooth vector fields $X$ and $Y$ are
topologically conjugate if there is a homeomorphism $h$ on $M$
takes the trajectory of $X$ at a given point $p\in M$ to the
trajectory of $Y$ at $h(p)$ and it preserves the direction of the
trajectories.
\begin{theorem}\label{th1}
Let $X$ be a connection stable vector field, and $c$ be a nonzero
real number. Then the vector field $cX$ is a connection stable
vector field.
\end{theorem}
\begin{proof}
The Equation \ref{eq4} implies that
$\Gamma^{jcX}_{ik}(\alpha^{cX}(t))=\frac{-\frac{dcX^{j}}{dt}}{cX^{k}(t)cX^{i}(t)}=\frac{1}{c}\Gamma^{jX}_{ik}(\alpha^{X}(t)).
 $ There is  $|c|\epsilon >0$ such that if
 $\frac{1}{c}Y\in V(M)$ with $|S(\frac{1}{c}Y)|=|S(X)|$
 and $|\Gamma^{j\frac{1}{c}Y}_{ik}(r)-\Gamma^{jX}_{ik}(r)|<|c|\epsilon
 $ for all
 $r\in M\setminus (S(\frac{1}{c}Y)\cup S(X))$, then $\frac{1}{c}Y$ is topologically conjugate to $X$.
 For given $Y\in V(M)$ with
 $|\Gamma^{jY}_{ik}(r)-\Gamma^{jcX}_{ik}(r)|<\epsilon $ for $r\in M\setminus (S(Y)\cup
 S(cX))$ we have
 $$|\Gamma^{jY}_{ik}(r)-\Gamma^{jcX}_{ik}(r)|=\frac{1}{|c|}|\Gamma^{j\frac{1}{c}Y}_{ik}(r)-\Gamma^{jX}_{ik}(r)|<\epsilon=\frac{1}{|c|}|c|\epsilon
 $$ for all
 $r\in M\setminus (S(Y)\cup S(cX)).$ Thus
 $|\Gamma^{j\frac{1}{c}Y}_{ik}(r)-\Gamma^{jX}_{ik}(r)|<|c|\epsilon$.
 Hence $\frac{1}{c}Y$ is topologically conjugate to $X$, and this
 implies that  $Y$ is topologically conjugate to $cX$.
\end{proof}

In the next theorem we show that in the case of compact manifold,
the structural stability implies  to connection stability and vice
versa.
\begin{theorem}\label{th2}
 If $M$ is a smooth compact manifold, then a smooth vector field
 $X$ on $M$ is structurally stable if and only if $X$ is
 connection stable.
\end{theorem}
\begin{proof}
Suppose $X$ is a structurally stable vector field on $M$ which is
not connection stable. Then for each natural number $n$ there is a
smooth vector field $Y_{n}$ with $|S(Y_{n})|=|S(X)|$, and for each
$r\in M\setminus (S(Y_{n})\cup S(X)) $,
$|\Gamma^{jY_{n}}_{ik}(r)-\Gamma^{jX}_{ik}(r)|<\frac{1}{n}$ such
that $Y_{n}$ is not topologically conjugate to $X$. Since
\begin{equation}\frac{d
{Y_{n}}^{j}}{dt}(t)+{Y_{n}}^{k}(t){Y_{n}}^{i}(t)\Gamma^{jY_{n}}_{ik}(\alpha^{Y_{n}}
(t))=0,
\end{equation}
 and
 $\Gamma^{jY_{n}}_{ik}(\alpha^{Y_{n}}(t))\rightarrow
 \Gamma^{jX}_{ik}(\alpha^{X}(t))$ when $n\rightarrow \infty $, then
 the sequence $\{Y_{n}(t)\}$ tends to the solution of the system \begin{equation}\frac{d
 X^{j}}{dt}(t)+{X}^{k}(t){X}^{i}(t)\Gamma^{jX}_{ik}(\alpha^{X}
(t))=0.
\end{equation}
By the uniqueness theorem of O.D.E. \cite{A}  we have
$Y_{n}\rightarrow X$ when $n\rightarrow \infty.$ The structural
stability of $X$ implies that, there is $n_{0}$  such that if
$n>n_{0}$, then $Y_{n}$ is topologically conjugate to $X$, which
is a contradiction. Thus $X$ is a connection stable vector field
on
$M$.\\
To prove the converse we assume $X$ is a connection stable vector
field, and $X$ is not a structurally stable vector field. Since
the number of singularities of $X$ is finite, then there is a
sequence $\{Y_{n}(t)\}$ of smooth vector fields  with
$|S(Y_{n})|=|S(X)|$ such that $Y_{n}\rightarrow X$ when
$n\rightarrow \infty$, and $Y_{n}$ is not topologically conjugate
to $X$ for each $n$. Since
\begin{equation}\frac{d
{Y_{n}}^{j}}{dt}(t)+{Y_{n}}^{k}(t){Y_{n}}^{i}(t)\Gamma^{jY_{n}}_{ik}(\alpha^{Y_{n}}
(t))=0,
\end{equation}
and
\begin{equation}\frac{d
 X^{j}}{dt}(t)+{X}^{k}(t){X}^{i}(t)\Gamma^{jX}_{ik}(\alpha^{X}
(t))=0,
\end{equation}
then $\Gamma^{jY_{n}}_{ik}(\alpha^{Y_{n}} (t))\rightarrow
\Gamma^{j}_{ik}(\alpha^{X} (t))$. The compact property of $M$
implies that the former convergence is a uniform convergence. So
$\Gamma^{jY_{n}}_{ik} \rightarrow \Gamma^{jX}_{ik}$ uniformly. The
connection stability of $X$ implies that there is $n_{0}$ such
that for $n>n_{0}$, $Y_{n}$ is topologically conjugate to $X$,
which is a contradiction. Hence $X$ is structurally stable.
\end{proof}
\section{EXAMPLES}
We begin this section by presenting an example of a connection
stable vector field.
\begin{example}
Let $M$ be the open interval $(1, 50)$ as an open submanifold of
$\mathbb{R}$. The vector field $X(x)=2x$ is a connection stable
vector field on $M$. Because if we take $\epsilon=\frac{1}{100}$,
then for each vector field $Y(x)=f(x)$ with
$|\Gamma^{1X}_{11}(x)-\Gamma^{1Y}_{11}(x)|<\frac{1}{2}$ we have
$|\Gamma^{1X}_{11}(x)-\Gamma^{1Y}_{11}(x)|=|\frac{-1}{x}+\frac{f'(x)}{f(x)}|<\frac{1}{100}$.
Thus $\frac{1}{100}<\frac{f'(x)}{f(x)}$. Hence we have two
cases:\\
Case 1. $f(x)>0$ and $f'(x)>0$.\\
Case 2.  $f(x)<0$ and $f'(x)<0$.\\
In these two cases two vector fields are topologically conjugate.
\end{example}
Now we show that harmonic oscillator \cite{N} is not a connection
stable vector field.
\begin{example}
Let $M$ be the interior of the circle of radius $2$ with the
center $(0,0)$ in $\mathbb{R}^2$ as an open submanifold of
$\mathbb{R}^2$. We show that the vector field $X(x,y)=(y,-x)$ is
not a connection stable vector
field on $M$. The Christoffel symbols of this vector fields are:\\
$$\Gamma^{1X}_{11}(x,y)=\frac{x}{y^{2}},~
\Gamma^{1X}_{12}(x,y)=\Gamma^{1X}_{21}(x,y)=\frac{-1}{y}, ~
\Gamma^{1X}_{22}(x,y)=\frac{1}{x},$$
$$\Gamma^{2X}_{11}(x,y)=\frac{1}{y},~\Gamma^{2X}_{12}(x,y)=\Gamma^{2X}_{21}(x,y)=\frac{-1}{x},
~ \Gamma^{2X}_{22}(x,y)=\frac{y}{x^{2}}.$$ For $\delta
>0$ if we take $Y(x,y)=(y+\delta x, -x)$, then it's Christoffel symbols
are:
$$\Gamma^{1Y}_{11}(x,y)=\frac{x-\delta y-\delta^{2}x}{(y+\delta x)^{2}},$$ $$
\Gamma^{1Y}_{12}(x,y)=\Gamma^{1Y}_{21}(x,y)=\frac{x-\delta
y-\delta^{2}x}{-x(y+\delta x)},
$$
$$ \Gamma^{1Y}_{22}(x,y)=\frac{x-\delta y-\delta^{2}x}{x^{2}},~
\Gamma^{2Y}_{11}(x,y)=\frac{1}{y+\delta x},$$
$$\Gamma^{2Y}_{12}(x,y)=\Gamma^{2Y}_{21}(x,y)=\frac{-1}{x},$$
$$  \Gamma^{2Y}_{22}(x,y)=\frac{y+\delta x}{x^{2}}.$$
Now for a given $\epsilon >0$ we can take $\delta$ sufficiently
small so that
$$|\Gamma^{jY}_{ik}(x,y)-\Gamma^{jX}_{ik}(x,y)|<\epsilon.$$
But the vector fields $X$ and $Y$ are not topologically conjugate,
because $X$ has periodic orbits, but $Y$ has no any periodic
orbit.
\end{example}
   In the next example we present a technique for proving
   unstability of a vector field.
\begin{example}\label{Ex3}
The vector field $Y(x,y,z)=(y,-x, x)$ on $R^{3}$ is not a
connection stable vector field.\\
We have:
$$\Gamma^{1Y}_{11}(x,y,z)=\frac{x}{y^{2}},~\Gamma^{1Y}_{12}(x,y,z)=\frac{-1}{y},~\Gamma^{1Y}_{13}(x,y,z)=\frac{1}{y},$$
$$\Gamma^{1Y}_{22}(x,y,z)=\frac{1}{x},~\Gamma^{1Y}_{23}(x,y,z)=\frac{-1}{x},~\Gamma^{1Y}_{33}(x,y,z)=\frac{1}{x},$$
$$\Gamma^{2Y}_{11}(x,y,z)=-\Gamma^{3Y}_{11}(x,y,z)=\frac{1}{y},~\Gamma^{2Y}_{12}(x,y,z)=-\Gamma^{3Y}_{12}(x,y,z)=\frac{-1}{x},$$
$$\Gamma^{2Y}_{13}(x,y,z)=-\Gamma^{3Y}_{13}(x,y,z)=\frac{1}{x},~\Gamma^{2Y}_{22}(x,y,z)=-\Gamma^{3Y}_{22}(x,y,z)=\frac{y}{x^{2}},$$
$$\Gamma^{2Y}_{23}(x,y,z)=-\Gamma^{3Y}_{23}(x,y,z)=\frac{y}{-x^{2}},~\Gamma^{2Y}_{33}(x,y,z)=-\Gamma^{3Y}_{33}(x,y,z)=\frac{y}{x^{2}}.$$
 If we take $Z(x,y,z)=(x,y,x)$, then we have:
$$\Gamma^{1Z}_{11}(x,y,z)=\Gamma^{3Z}_{11}(x,y,z)=\frac{1}{x},~\Gamma^{1Z}_{12}(x,y,z)=\Gamma^{3Z}_{12}(x,y,z)=\frac{-1}{y},$$
$$\Gamma^{1Z}_{13}(x,y,z)=\Gamma^{3Z}_{13}(x,y,z)=\frac{-1}{x},~\Gamma^{1Z}_{33}(x,y,z)=\Gamma^{3Z}_{33}(x,y,z)=\frac{-1}{x},$$
$$\Gamma^{2Z}_{11}(x,y,z)=\frac{-y}{x^{2}},~\Gamma^{2Z}_{12}(x,y,z)=\frac{-1}{x},~\Gamma^{2Z}_{13}(x,y,z)=\Gamma^{2Z}_{33}(x,y,z)=\frac{-y}{x^{2}}.$$
For a given $\epsilon >0$ there is $B>0$ such that
$|\Gamma^{jY}_{ik}(x,y,z)|<\frac{\epsilon}{4},$ and
$|\Gamma^{jZ}_{ik}(x,y,z)|<\frac{\epsilon}{4},$ for $2B>|x|>B$ and
$2B>|y|>B$. If  $Y$ is connection stable, then  it must be
connection stable on $D=\{(x,y,z)\in R^{3}~:~ 2B>|x|>B,~and ~
2B>|y|>B\}$. We see that on $D$, $ |\Gamma^{iY}_{jk}-
\Gamma^{iZ}_{jk}|<\epsilon $, but these two vector fields are not
topologically conjugate, because $Z$ has no any periodic orbit but
$Y$ has many periodic orbits of the form
$O(a,b,-b)=\{(acos(t)+bsin(t), -asin(t)+bcos(t),
asin(t)-bcos(t))~:~ t\in  R\},$ when
$\sqrt{2}B<\sqrt{a^{2}+b^{2}}<2\sqrt{2}B.$
\end{example}
\begin{theorem}\label{Th2}
Let $X$ be a vector field on $R^{3}$ such that for a given
$\epsilon
>0$ there is $A>0$ so that for each $B>A$ we have $|\Gamma^{jX}_{ik}(x,y,z)|<\frac{\epsilon}{4},$ for $2B>|x|>B$ and
$2B>|y|>B$. Then $X$ is not a connection stable vector field.
\end{theorem}
\begin{proof}
We take $Y$ and $Z$ as two vector fields which are presented in
Example \ref{Ex3}. For a given $\epsilon >0$ we take $B>A$ large
enough so that $|\Gamma^{jY}_{ik}(x,y,z)|<\frac{\epsilon}{4},$
$|\Gamma^{jZ}_{ik}(x,y,z)|<\frac{\epsilon}{4},$ and
$|\Gamma^{jX}_{ik}(x,y,z)|<\frac{\epsilon}{4},$ for $2B>|x|>B$ and
$2B>|y|>B$. If $X$ be a connection stable vector field, then it
must be connection stable on $D=\{(x,y,z)\in R^{3}~:~
2B>|x|>B,~and ~ 2B>|y|>B\}$. So it is topologically conjugate to
$Y$ and $Z$ on $D$. Thus $Y$ and $Z$ are topologically conjugate
on $D$, which is a contradiction. Thus $X$ is not a connection
stable vector field.
\end{proof}
\section{Einstein-Vlasov System and galaxies rotation curves}
In this section we consider the stationary solutions of the galaxy
model which is presented in \cite{CO}. The problem of the galaxies
rotation curves is a portion of the more general Dark Matter
problem. It consists in the strong discrepancy between the
observed experimental curves and the curve derived by applying the
theory of Newtonian gravity to the matter observed in the galaxies
\cite{Rubin}. The most popular way of solving the problem is to
postulate the existence of a hypothetical component of matter
which, unlike known matter, would not emit electromagnetic
radiation and would currently only be detectable indirectly
through its gravitational effects \cite{Salucci}. An alternative
approach is to extend general relativity and, correspondingly, its
Newtonian approximation, in the well-known framework of extended
theories of gravity \cite{Corda}. In \cite{CO} it has been shown
that, in the framework of $R^2$ gravity  and in the linearized
approach, one can obtain spherically symmetric stationary states
which can be used as a model for galaxies. This approach could be,
in principle, a solution to the problem of the galaxies rotation
curves. In this framework, which is founded on the Einstein-Vlasov
system, the Ricci curvature  generates a high energy term that, in
principle, can be identiﬁed as the dark matter ﬁeld making up
the galaxy. In particular, the Authors of \cite{CO} discussed the
solutions of a Klein-Gordon equation for the spacetime curvature.
Such solutions describe high energy scalarons, which are
gravitational ﬁelds that in the context of galactic dynamics can
be interpreted like the no-light-emitting galactic component. The
solutions are \cite{CO}:
\begin{equation}\label{Eq3}\frac{1}{r^{2}}\frac{d}{dr}(\frac{d}{dr}b\widetilde{R}(t,z)r^{2})=(1+2b\widetilde{R}(t,z))\mu(r)\end{equation}
\begin{equation}\mu(r)=\int{\frac{dp}{\sqrt{1+||p||^{2}}}f(X,p)}\end{equation}
\begin{equation}p.\partial_{X}f(X,p)-\frac{1}{r}\frac{d}{dr}b\widetilde{R}(t,z)[(p.X)p+X].\partial_{p}f(X,p)=0,\end{equation}
where $X=(x,y,z)$, $p=(p_{1},p_{2},p_{3})$, and $r=||X||$.\\
If we take $v=b\widetilde{R}r^{2}$, and $u=\frac{dv}{dr}$ then we
deduce the following system from Equation \ref{Eq3}.\\
\begin{equation}\label{Eq4}
          \left\{
         \begin{array}{rl}
         \frac{dv}{dr}=u \hspace{1.85cm}\\
        \frac{du}{dr}= (r^{2}+2v)\mu(r)
         \end{array} \right..
         \end{equation}
So we have $\frac{d^{2}v}{dr^{2}}=\frac{du}{dr}= (r^{2}+2v)\mu(r)$
and
$$\frac{d^{2}u}{dr^{2}}=(2r+2\frac{dv}{dr})\mu(r)+(r^{2}+2v)\mu'(r)$$
$$=(2r+2u)\mu(r)+(r^{2}+2v)\mu'(r).$$ The Christoffel symbols of
the system \ref{Eq4} are:\\
$$\Gamma^{1}_{11}=\frac{-(r^{2}+2v)\mu(r)}{u^{2}},~
\Gamma^{1}_{12}=\Gamma^{1}_{21}=\frac{-1}{u},$$
$$\Gamma^{1}_{22}=\frac{-1}{(r^{2}+2v)\mu(r)},~
\Gamma^{2}_{11}=\frac{-(2r+2u)\mu(r)-(r^{2}+2v)\mu'(r)}{u^{2}},$$
$$\Gamma^{2}_{12}=\Gamma^{2}_{21}=\frac{-(2r+2u)}{u(r^{2}+2v)}-\frac{\mu'(r)}{u\mu(r)},$$
$$\Gamma^{2}_{22}=\frac{-(2r+2u)}{(r^{2}+2v)^{2}\mu(r)}-\frac{\mu'(r)}{(r^{2}+2v)\mu^{2}(r)}.$$
Now we show that the non-autonomous system \ref{Eq4} is not a
connection stable system. In fact for a given $\epsilon >0$ by
choosing $r_{0}$ large enough and restricting ourself to the
region $r_{0}<r<2r_{0}$ we can rich to the condition
$|\Gamma^{j}_{ik}(t,r)|<\frac{\epsilon}{4}.$ In this region we
have a harmonic oscillator with the Christoffel symbols arbitrary
close to the Christoffel symbols of system \ref{Eq4}. Since
harmonic oscillator is not connection stable, then system
\ref{Eq4} is not
a connection stable system.\\
\section{Conclusion}
In this work we extend the notion of stability of vector fields on
compact manifolds to the vector fields on noncompact manifold. In
section $4$ we extend the notion of connection stability  of
vector fields or autonomous systems to non-autonomous systems. In
fact Definition \ref{def1} is extendable for non-autonomous
systems. As a concrete physical case we considered the
Einstein-Vlasov system analyzed in \cite{CO}, where, in the
framework of $R^2$ gravity and in the linearized approach, it is
possible to obtain spherically symmetric stationary states which
are used as a model for galaxies. The model, which is founded on
the Einstein-Vlasov system could,  in principle, solve the Dark
Matter Problem. In fact, the Ricci curvature generates a high
energy field, the so-called "scalaron", which, can be identified
as the dark matter field constructing the galaxy. In particular,
the solutions of a  Klein-Gordon equation for the spacetime
curvature are discussed in \cite{CO}. Such solutions represent
high energy scalarons, which are particular gravitational fields,
that in the context of galactic  dynamics can be interpreted like
the no-light-emitting galactic component. In that case,
non-autonomous system which has been obtained did not result a
connection stable system. The consideration of connection
stability for
non-autonomous systems can be a topic for further research.\\
 \section{Data  availability:}  
No Data associated in the manuscript.


\begin{thebibliography}{9}
\bibitem{A1} S.A. Alkharashi, W. Alotaibi, Structural stability of a porous channel of electrical flow affected by periodic velocities. J Eng Math 141, 3 (2023).
https://doi.org/10.1007/s10665-023-10278-3.
\bibitem{A} V.I. Arnold, Ordinary Differential Equations. Springer, Berlin,
(1992).
\bibitem{B} C.D. Bailey, Hamilton's law and the stability of nonconservative continuous systems, AIAA Journal 18, 3, 347-96349 (1980).
\bibitem{B1} C. Biezeno, H. Hencky, On the general theory of elastic stability, Proc. Acad. Sci. Amsterdam 31, 6
(1928).
\bibitem{C} S. Chandrasekhar, Hydrodynamic and hydromagnetic stability. Oxford Univ. Press,
Oxford, 1961.
\bibitem{Salucci} E. Corbelli and P. Salucci, The extended rotation curve and the dark matter halo of M33, MNRAS 311, 441 (2000).
\bibitem{Corda} C. Corda, Interferometric detection of gravitational waves: the definitive test for general relativity, Int. Jour. Mod. Phys. D 18, 2275 (2009).
\bibitem{CO} C. Corda, H.J. Mosquera Cuesta, R. Lorduy Gomez, High-energy scalarons in $R^{2}$ gravity as a model for Dark Matter in
galaxies, Astroparticle Physics 35 (2012) 362-96370.
\bibitem{MG} M.R. Molaei, The geometry of trajectories, International Journal of Geometric Methods in Modern
Physics, (2020) 2050051 (18 pages).
\bibitem{N} G. Naber, The simple harmonic oscillator: an introduction to the mathematics
of quantum mechanics, Springer 2015.
\bibitem{PD} J. Palis, W. de Melo, Geometric theory of dynamical
systems, Springer-Verlag, 1982.
\bibitem{Rubin} V. C. Rubin and W. K. Ford Jr, ApJ 159, 379 (1970).
\bibitem{S} R.T. Shield, A.E. Green, On certain methods in the stability theory of continuous systems, Arch. Rat. Mech. Anal. 12, 4, 354-96360 (1963).




\end{thebibliography}
\end{document}